\documentclass{llncs} 

\usepackage{amssymb,amsmath,amscd,latexsym,epic,eepic}
\usepackage{algorithm} 
\usepackage{gastex}
\usepackage{times}
\usepackage{paralist}

\def\qed{\rule{0.4em}{1.4ex}} 
\newcommand{\restr}{\upharpoonright} \newcommand{\set}[1]{\{#1\}}

\newcommand{\outcome}{\mathrm{Outcome}}
\newcommand{\Prb}{\mathrm{Pr}}
\newcommand{\Exp}{\mathbb{E}}
\newcommand{\seq}[1]{\langle #1 \rangle}

\newcommand{\trans}{\delta}
\newcommand{\distr}{{\cal D}}

\newcommand{\supp}{\mathrm{Supp}}
\newcommand{\Supp}{\mathrm{Supp}}

\newcommand{\mem}{{\tt M}}

\newcommand{\vare}{\varepsilon}

\newcommand{\vdisc}{\vec{\lambda}}
\newcommand{\disc}{\lambda}
\newcommand{\exit}{\mathsf{ex}}

\newcommand{\slopefrac}[2]{\leavevmode\kern.1em
  \raise .5ex\hbox{\the\scriptfont0 #1}\kern-.1em
  /\kern-.15em\lower .25ex\hbox{\the\scriptfont0 #2}}
\newcommand{\half}{\slopefrac{1}{2}}

\newcommand{\pat}{\omega} 
\newcommand{\pats}{\Omega} 
 \newcommand{\Paths}{\Omega}

\newcommand{\dist}{\mathit{dist}_R}
\newcommand{\dista}{\mathit{dist}_A}
\newcommand{\equivclass}[1]{{[\![ #1 ]\!]_\equiv}}

 \newcommand{\Nats}{\mathbb{N}}
\newcommand{\reals}{\mathbb{R}} 
\newcommand{\Inf}{\mathrm{Inf}}

\newcommand{\ov}{\overline}

\newcommand{\calf}{{\mathcal F}}
\newcommand{\MDT}{\mathsf{MDT}}
\newcommand{\MT}{\mathsf{MT}}
\newcommand{\indi}{\mathbf{1}}

\newcommand{\Val}{{\mathsf{Val}}}
\newcommand{\game}{G}
\newcommand{\mov}{\Gamma}
\newcommand{\moves}{A}
\newcommand{\dest}{\mathrm{Dest}}
\def\bigstra{\Pi}
\def\stra{\pi}
\newcommand{\cala}{{\mathcal A}}
\newcommand{\ParityCond}{{\text{\textrm{Parity}}}(p)}

\makeatletter

\begingroup \catcode `|=0 \catcode `[= 1
\catcode`]=2 \catcode `\{=12 \catcode `\}=12
\catcode`\\=12 |gdef|@xcomment#1\end{comment}[|end[comment]]
|endgroup

\def\@comment{\let\do\@makeother \dospecials\catcode`\^^M=10\def\par{}}

\def\begincomment{\@comment\@xcomment}

\makeatother

\newenvironment{comment}{\begincomment}{}

\begin{document}
\title{ 
Robustness of Structurally Equivalent \\ Concurrent Parity Games\thanks{The research 
was supported by 
Austrian Science Fund (FWF) Grant No P 23499-N23,
FWF NFN Grant No S11407-N23 (RiSE), ERC Start grant (279307: Graph Games), 
and Microsoft faculty fellows award.}
}
\author{%
Krishnendu Chatterjee
}
\institute{%
IST Austria (Institute of Science and Technology Austria)
}
\maketitle

\vspace{-1em}
\begin{abstract} 
We consider two-player stochastic games played on a finite state space for an 
infinite number of rounds.
The games are {\em concurrent}: in each round, the two players 
(player~1 and player~2) choose their moves independently and 
simultaneously; the current state and the two moves determine a 
probability distribution over the successor states. 
We also consider the important special case of turn-based stochastic
games where players make moves in turns, rather than concurrently.
We study concurrent games with $\omega$-regular winning conditions specified as 
{\em parity\/} objectives.
The value for player~1 for a parity objective is the maximal probability with which the 
player can guarantee the satisfaction of the objective against all
strategies of the opponent. 
We study the problem of continuity and robustness of the value function 
in concurrent and turn-based stochastic parity games with respect to 
imprecision in the transition probabilities.
We present quantitative bounds on the difference of the value function (in terms
of the imprecision of the transition probabilities) and  
show the value continuity for structurally equivalent concurrent games 
(two games are structurally equivalent if the supports of the transition
functions are the same and the probabilities differ). 
We also show robustness of optimal strategies for structurally equivalent
turn-based stochastic parity games.
Finally, we show that the value continuity property breaks without the
structural equivalence assumption (even for Markov chains) and show
that our quantitative bound is asymptotically optimal.
Hence our results are tight (the assumption is both necessary and sufficient)
and optimal (our quantitative bound is asymptotically optimal).
\end{abstract}

\section{Introduction}

Concurrent stochastic games are played by two players on a finite state space for an 
infinite number of rounds.
In every round, the two players simultaneously and independently choose moves (or actions), 
and the current state and the two chosen moves determine a probability 
distribution over the successor states. 
The outcome of the game (or a \emph{play}) is an infinite sequence of states.
These games were introduced by Shapley~\cite{Sha53}, and have been one of 
the most fundamental and well studied game models in stochastic graph games.
We consider $\omega$-regular objectives specified as parity objectives;
that is, given an $\omega$-regular set $\Phi$ of infinite state sequences,
player~1 wins if the outcome of the game lies in~$\Phi$.
Otherwise, player~2 wins, i.e., the game is zero-sum.
The class of concurrent stochastic games subsumes many other important classes 
of games as sub-classes: (1) \emph{turn-based stochastic} games, where in 
every round only one player chooses moves (i.e., the players make moves in 
turns); and (2) \emph{Markov decision processes (MDPs)} (one-player stochastic games).
Concurrent games and the sub-classes provide a rich framework to model 
various classes of dynamic reactive systems, and $\omega$-regular
objectives provide a robust specification language to express all
commonly used properties in verification, and all $\omega$-regular 
objectives can be expressed as parity objectives. 
Thus concurrent games with parity objectives provide the mathematical
framework to study many important problems in the synthesis and 
verification of reactive systems~\cite{Church62,RamadgeWonham87,PnueliRosner89} 
(see also~\cite{ALW89,Dill89book,AHK02}).

The player-1 \emph{value} $v_1(s)$ of the game at a state $s$ is the 
limit probability with which player~1 can ensure that the outcome of 
the game lies in~$\Phi$;
that is, the value $v_1(s)$ is the maximal probability with which 
player~1 can guarantee $\Phi$ against all strategies of player~2.
Symmetrically, the player-2 \emph{value} $v_2(s)$ is the limit 
probability with which player~2 can ensure that the outcome of 
the game lies outside~$\Phi$.
The problem of studying the computational complexity of 
MDPs, turn-based stochastic games, and concurrent games with 
parity objectives has received a lot of attention in literature.
Markov decision processes with $\omega$-regular objectives
have been studied in~\cite{CY95,deAlfaro97,CHJ04} and the results show
existence of pure (deterministic) memoryless (stationary) optimal strategies 
for parity objectives and the problem of value computation is achievable 
in polynomial time.
Turn-based stochastic games with the special case of reachability objectives
have been studied in~\cite{Con92} and existence of pure memoryless
optimal strategies has been established and the decision problem 
of whether the value at a state is at least a given rational value 
lies in NP $\cap$ coNP.
The existence of pure memoryless optimal strategies for turn-based 
stochastic games with parity objectives was established in~\cite{CJH04,Zie04},
and again the decision problem lies in NP $\cap$ coNP.
Concurrent parity games have been studied in~\cite{dAH00,dAM01,CdAH04b,EY06} and
for concurrent parity games optimal strategies need not exist, and 
$\vare$-optimal strategies (for $\vare>0$) require both infinite memory and 
randomization in general, and the decision problem can be solved in PSPACE.

Almost all results in the literature consider the problem of computing
values and optimal strategies when the game model is given precisely 
along with the objective. 
However, it is often unrealistic to know the precise probabilities of 
transition which are only estimated through observation. 
Since the transition probabilities are not known precisely, an extremely
important question is how robust is the analysis of concurrent games
and its sub-classes with parity objectives with respect to small changes 
in the transition probabilities. 
This question has been largely ignored in the study of concurrent and 
turn-based stochastic parity games.
In this paper we study the following 
problems related to continuity and robustness of values: 
(1)~\emph{(continuity of values):} under what conditions can continuity of the value function be 
proved for concurrent parity games; 
(2)~\emph{(robustness of values):} can quantitative bounds be obtained 
on the difference of the value function in terms of the difference of the 
transition probabilities; and 
(3)~\emph{(robustness of optimal strategies):} 
do optimal strategies of a game remain $\vare$-optimal,
for $\vare>0$, if the transition probabilities are slightly changed.

\smallskip\noindent{\em Our contributions.} Our contributions are as follows:
\begin{compactenum}
\item We consider \emph{structurally equivalent} game structures, where the
supports of the transition probabilities are the same, but the precise 
transition probabilities may differ. 
We show the following results for structurally equivalent concurrent parity 
games:
\begin{compactenum}
\item \emph{Quantitative bound.} We present a quantitative bound on the 
difference of the value functions of two structurally equivalent game 
structures in terms of the difference of the transition probabilities. 
We show when the difference in the transition probabilities are small, our 
bound is asymptotically optimal.
Our example to show the matching lower bound is on a Markov chain, and thus our
result shows that the bound for a Markov chain can be generalized to concurrent
games.

\item \emph{Value continuity.} We show \emph{value continuity} for structurally
equivalent concurrent parity games, i.e., as the difference in the transition 
probabilities goes to~0, the difference in value functions also goes to~0.
We then show that the structural equivalence assumption is necessary: we show a 
family of Markov chains (that are not structurally equivalent) where the 
difference of the transition probabilities goes to~0, but the difference 
in the value functions is~1.
It follows that the structural equivalence assumption is both necessary (even 
for Markov chains) and sufficient (even for concurrent games).
\end{compactenum}
It follows from above that our results are both optimal (quantitative bounds) as well as 
tight (assumption both necessary and sufficient).
Our result for concurrent parity games is also a significant quantitative generalization of 
a result for concurrent parity games of~\cite{dAH00} which shows that the set of states
with value~1 remains same if the games are structurally equivalent.
We also argue that the structural equivalence assumption is not unrealistic in many cases:
a reactive system consists of many state variables, and given a state (valuation of variables)
it is typically known which variables are possibly updated, and what is unknown is 
the precise transition probabilities (which are estimated by observation). 
Thus the system that is obtained for analysis is structurally equivalent to the underlying
original system and it only differs in precise transition probabilities.

\item For turn-based stochastic parity games the value continuity and the quantitative bounds 
are same as for concurrent games. We also prove a stronger result for structurally equivalent 
turn-based stochastic games that shows that along with continuity of the value function, there
is also robustness property for pure memoryless optimal strategies. 
More precisely, for all $\vare>0$, we present a bound $\beta>0$, such that any 
pure memoryless optimal strategy in a turn-based stochastic parity game is an $\vare$-optimal
strategy in every structurally equivalent turn-based stochastic game such that the transition 
probabilities differ by at most $\beta$. 
Our result has deep significance as it allows the rich literature of work on turn-based 
stochastic games to carry over robustly for structurally equivalent turn-based stochastic games.
As argued before the model of turn-based stochastic game obtained to analyze may differ  
slightly in precise transition probabilities, and our results shows that the analysis 
on the slightly imprecise model using the classical results carry over to the underlying 
original system with small error bounds.

\end{compactenum}
Our results are obtained as follows. 
The result of~\cite{dAHM03} shows that the value function for 
concurrent parity games can be characterized as the limit of the 
value function of concurrent multi-discounted games (concurrent 
discounted games with different discount factors associated with 
every state). 
There exists bound on difference on value function of discounted games~\cite{FV97},
however, the bound depends on the discount factor, and in the limit 
gives trivial bounds (and in general this approach does not work as 
value continuity cannot be proven in general and the structural equivalence
assumption is necessary).
We use a classical result on Markov chains by Friedlin and Wentzell~\cite{FW84} and 
generalize a result of Solan~\cite{Sol03} from Markov chains with single discount to
Markov chains with multi-discounted objective to obtain a bound that is 
independent of the discount factor for structurally equivalent games.
Then the bound also applies when we take the limit of the discount factors, and
gives us the desired bound.

Our paper is organized as follows: in Section~\ref{sec:defn} we present the basic definitions,
in Section~\ref{sec:mc} we consider Markov chains with multi-discounted and parity objectives;
in Section~\ref{sec:games} (Subsection~\ref{sec:tb}) we prove the results related to turn-based 
stochastic games (item (2) of our contributions) and finally in Subsection~\ref{sec:conc} we 
present the quantitative bound and value continuity for concurrent games along with
the two examples to illustrate the asymptotic optimality of the bound and the 
structural equivalence assumption is necessary.
Detailed proofs 
are presented in the appendix. 

\vspace{-1em}
\section{Definitions}\label{sec:defn}

In this section we define game structures, strategies, objectives,  
values and present other preliminary definitions.

\smallskip\noindent{\bf Probability distributions.}
For a finite set~$A$, a {\em probability distribution\/} on $A$ is a
function $\trans\!:A\mapsto[0,1]$ such that $\sum_{a \in A} \trans(a) = 1$.
We denote the set of probability distributions on $A$ by $\distr(A)$. 
Given a distribution $\trans \in \distr(A)$, we denote by $\supp(\trans) = 
\{x\in A \mid \trans(x) > 0\}$ the {\em support\/} of the distribution 
$\trans$.

\smallskip\noindent{\bf Concurrent game structures.} 
A (two-player) {\em concurrent stochastic game structure\/} 
$\game = \langle S, \moves,\mov_1, \mov_2, \trans \rangle$ consists of the 
following components.

\begin{itemize}

\item A finite state space $S$ and a finite set $A$ of moves (or actions).

\item Two move assignments $\mov_1, \mov_2 \!: S\mapsto 2^\moves
	\setminus \emptyset$.  For $i \in \{1,2\}$, assignment
	$\mov_i$ associates with each state $s \in S$ the nonempty
	set $\mov_i(s) \subseteq \moves$ of moves available to player $i$
	at state $s$.  

\item A probabilistic transition function
	$\trans\!:S\times \moves \times\moves\mapsto \distr(S)$, which
	associates with every state $s \in S$ and moves $a_1 \in
	\mov_1(s)$ and $a_2 \in \mov_2(s)$ a probability
	distribution $\trans(s,a_1,a_2) \in \distr(S)$ for the
	successor state.
\end{itemize}

\noindent{\bf Plays.}
At every state $s\in S$, player~1 chooses a move $a_1\in\mov_1(s)$,
and simultaneously and independently
player~2 chooses a move $a_2\in\mov_2(s)$.  
The game then proceeds to the successor state $t$ with probability
$\trans(s,a_1,a_2)(t)$, for all $t \in S$. 
For all states $s \in S$ and moves $a_1 \in
\mov_1(s)$ and $a_2 \in \mov_2(s)$, we indicate by 
$\dest(s,a_1,a_2) = \supp(\trans(s,a_1,a_2))$ 
the set of possible successors of $s$ when moves $a_1$, $a_2$
are selected. 
A {\em path\/} or a {\em play\/} of $\game$ is an infinite sequence
$\pat = \langle s_0,s_1,s_2,\ldots\rangle$ of states in $S$ such that for all 
$k\ge 0$, there are moves $a^k_1 \in \mov_1(s_k)$ and $a^k_2 \in \mov_2(s_k)$
such that $s_{k+1} \in \dest(s_k,a^k_1,a^k_2)$.
We denote by $\pats$ the set of all paths. 
We denote by $\theta_i$ the random variable that denotes
the $i$-th state of a path.
For a play $\pat = \seq{s_0, s_1, s_2,\ldots} \in \pats$,   
we define $\Inf(\pat) = 
\set{s \in S \mid \mbox{$s_k = s$ for infinitely many $k \geq 0$}}$
to be the set of states that occur infinitely often in~$\pat$.


\smallskip\noindent{\bf Special classes of games.}
We consider the following special classes of concurrent games.
\begin{compactenum}
\item \emph{Turn-based stochastic games.} 
A game structure $\game$ is {\em turn-based stochastic\/} if at every
state at most one player can choose among multiple moves; that is, for
every state $s \in S$ there exists at most one $i \in \{1,2\}$ with
$|\mov_i(s)| > 1$. 

\item \emph{Markov decision processes.}
A game structure is a \emph{player-1 Markov decision process (MDP)} if for all 
$s \in S$ we have $|\mov_2(s)|=1$, i.e., only player~1 has choice of 
actions in the game.
Similarly, a game structure is a \emph{player-2 MDP} if for all 
$s \in S$ we have $|\mov_1(s)|=1$.

\item \emph{Markov chains.} A game structure is a Markov chain if for all 
$s \in S$ we have $|\mov_1(s)|=1$ and $|\mov_2(s)|=1$. Hence in a Markov
chain the players do not matter, and for the rest of the 
paper a Markov chain consists of a tuple $(S,\trans)$ where 
$\trans:S \mapsto \distr(S)$ is the probabilistic transition function.

\end{compactenum}

\noindent{\bf Strategies.}
A {\em strategy\/} for a player is a recipe that describes how to 
extend a play.
Formally, a strategy for player $i\in\{1,2\}$ is a mapping 
$\stra_i\!:S^+\mapsto\distr(\moves)$ that associates with every nonempty 
finite sequence $x \in S^+$ of states, 
representing the past history of the game, 
a probability distribution $\stra_i(x)$ used to select
the next move. 
The strategy $\stra_i$ can prescribe only moves that are available to 
player~$i$;
that is, for all sequences $x\in S^*$ and states $s\in S$, we require that
$\supp(\stra_i(x\cdot s)) \subseteq \mov_i(s)$.  
We denote by $\bigstra_i$ the set of all strategies for player $i\in\{1,2\}$.

Given a state $s\in S$ and two strategies $\stra_1\in\bigstra_1$ and
$\stra_2\in\bigstra_2$, we define
$\outcome(s,\stra_1,\stra_2)\subseteq\pats$ to be the set of paths
that can be followed by the game, when the game starts {f}rom $s$ and
the players use the strategies $\stra_1$ and~$\stra_2$.  Formally,
$\seq{s_0, s_1, s_2, \ldots} \in \outcome(s,\stra_1,\stra_2)$ if $s_0=s$ and
if for all $k\ge 0$ there exist moves $a^k_1 \in \mov_1(s_k)$ and
$a^k_2 \in \mov_2(s_k)$ such that
(i)~$\stra_1(s_0,\ldots,s_k)(a^k_1) > 0$; 
(ii)~$\stra_2(s_0,\ldots,s_k)(a^k_2) > 0$; and
(iii)~$s_{k+1} \in \dest(s_k,a^k_1,a^k_2)$. 
Once the starting state $s$ and the strategies $\stra_1$ and $\stra_2$
for the two players have been chosen, 
the probabilities of events are uniquely defined~\cite{VardiP85}, where an {\em
event\/} $\cala\subseteq\pats$ is a measurable set of
paths\footnote{To be precise, we should define events as
measurable sets of paths {\em sharing the same initial state,} and
we should replace our events with families of events, indexed by their
initial state. 
However, our (slightly) improper definition leads to 
more concise notation.}. 
For an event $\cala\subseteq\pats$, we denote by
$\Prb_s^{\stra_1,\stra_2}(\cala)$ the probability that a path belongs to 
$\cala$ when the game starts {f}rom $s$ and the players use the strategies 
$\stra_1$ and~$\stra_2$.

\smallskip\noindent{\bf Classification of strategies.} We consider the following
special classes of strategies.

\begin{compactenum}

\item \emph{(Pure).} A strategy $\stra$ is {\em pure (deterministic)\/} if for 
all $x \in S^+$ there exists $a \in \moves$ such that $\stra(x)(a) = 1$.
Thus, deterministic strategies are equivalent to functions $S^+ \mapsto
\moves$.

\item \emph{(Finite-memory).} Strategies in general are \emph{history-dependent} and
can be represented as follows:
let $\mem$ be a set called \emph{memory} to remember the history of plays (the set
$\mem$ can be infinite in general).
A strategy with memory can be described as a pair of functions:
(a) a \emph{memory update} function $\stra_{u}: S \times \mem \mapsto \mem$,
that given the memory $\mem$ with the information about the history and
the current state updates the memory; and
(b) a \emph{next move} function $\stra_{n}: S\times\mem \mapsto \distr(\moves)$
that given the memory and the current state specifies the next move of
the player.
A strategy is \emph{finite-memory} if the memory $\mem$ is finite.

\item \emph{(Memoryless).} A \emph{memoryless} strategy is independent of the history of play and 
only depends on the current state.
Formally, for a memoryless strategy $\stra$ we have $\stra(x\cdot s) =
\stra (s)$ for all $s \in S$ and all $x \in S^*$.
Thus memoryless strategies are equivalent to functions 
$S \mapsto \distr(\moves)$.

\item \emph{(Pure memoryless).}
A strategy is \emph{pure memoryless} if it is both pure and memoryless.
Pure  memoryless strategies neither use memory, nor use randomization
and are equivalent to functions $S \mapsto \moves$.
\end{compactenum}


\smallskip\noindent{\bf Qualitative objectives.} 
We specify \emph{qualitative} objectives for the players by providing 
the set of \emph{winning plays} $\Phi \subseteq \pats$ for each player.
In this paper we study only zero-sum games \cite{RagFil91,FV97}, 
where the objectives of the two players are complementary.
A general class of objectives are the Borel objectives~\cite{Kechris}. 
A \emph{Borel objective} $\Phi \subseteq S^\omega$ is a Borel set in the 
Cantor topology on~$S^\omega$. 
In this paper we consider \emph{$\omega$-regular objectives}, 
which lie in the first $2\half$ levels of the Borel hierarchy
(i.e., in the intersection of $\Sigma_3$ and~$\Pi_3$)~\cite{Thomas90}.
All $\omega$-regular objectives can be specified as parity objectives,
and hence in this work we focus on parity objectives, and they 
are defined as follows.
\begin{compactitem}
\item
  \emph{Parity objectives.}
  For $c,d \in \Nats$, we let $[c..d] = \set{c, c+1, \ldots, d}$. 
  Let $p : S \mapsto [0..d]$ be a function that assigns a \emph{priority}
  $p(s)$ to every state $s \in S$, where $d \in \Nats$.
  The \emph{Even parity objective} requires that the minimum 
  priority visited infinitely often is even. Formally, the set
  of winning plays is defined as
  $\ParityCond= 
  \set{\pat \in \pats \mid
  \min\big(p(\Inf(\pat))\big) \text{ is even }}$.
\end{compactitem}

\smallskip\noindent{\bf Quantitative objectives.}
\emph{Quantitative} objectives are measurable functions $f:\pats \mapsto \reals$.
We will consider \emph{multi-discounted} objective functions, as there is a 
close connection established between concurrent games with multi-discounted
objectives and concurrent games with parity objectives.
Given a concurrent game structure with state space $S$, let $\vdisc$ be a \emph{discount vector} 
that assigns for all $s \in S$ a discount factor $0<\disc(s)<1$ (unless otherwise mentioned
we will always consider discount vectors $\vdisc$ such that for all $s \in S$ we have
$0<\disc(s)<1$).
Let $r:S \mapsto \reals$ be a reward function that assigns a real-valued reward $r(s)$ to 
every state $s \in S$. 
The multi-discounted objective function $\MDT(\vdisc,r):\Paths \mapsto \reals$ maps  
every path to the mean-discounted reward of the path.
Formally, the function is defined as follows: for a path $\pat=s_0 s_1 s_2 \ldots$ we have  
\[
\MDT(\vdisc,r)(\pat)= \frac{\sum_{j=0}^\infty (\prod_{i=0}^j \disc(s_i) ) \cdot r(s_j) }
{\sum_{j=0}^\infty (\prod_{i=0}^j \disc(s_i) ) }.
\]
Also note that a parity objective $\Phi$ can be intepreted as a function 
$\Phi:\pats \mapsto \set{0,1}$ by simply considering the characteristic function
that assigns~1 to paths that belong to $\Phi$ and~0 otherwise.

\smallskip\noindent{\bf Values, optimality, $\vare$-optimality.}
Given an objective $\Phi$ which is a measurable function $\Phi:\pats \mapsto \reals$,
we define the \emph{value} for player~1 of game $G$ with objective $\Phi$ 
from the state $s \in S$ as 
$
  \Val(G,\Phi)(s) =
  \sup_{\stra_1\in\bigstra_1}\inf_{\stra_2\in\bigstra_2}
  \Exp_s^{\stra_1,\stra_2}(\Phi); 
$ 
i.e., the value is the maximal expectation with which player~1 can 
guarantee the satisfaction of $\Phi$ against all player~2 strategies.
Given a player-1 strategy $\stra_1$, we use the notation
$
\Val^{\stra_1}(G,\Phi)(s)
= \inf_{\stra_2 \in \bigstra_2} \Exp_s^{\stra_1,\stra_2}(\Phi).
$
A strategy $\stra_1$ for player~1 is {\em optimal\/} for an
objective $\Phi$ if for all 
states $s \in S$, we have 
$
\Val^{\stra_1}(G,\Phi)(s)= \Val(G,\Phi)(s). 
$
For $\vare > 0$, a strategy $\stra_1$ for player~1 is 
{\em $\vare$-optimal\/} if for all states $s \in S$, we have 
$
\Val^{\stra_1}(G,\Phi)(s)
  \geq \Val(G,\Phi)(s) - \vare. 
$
The notion of values, optimal and $\vare$-optimal strategies for player~2 are
defined analogously.
The following theorem summarizes the results in literature related
to determinacy and memory complexity of concurrent games and its
sub-classes for parity and multi-discounted objectives.

\begin{theorem}\label{thrm_lit1}
The following assertions hold:
\begin{enumerate}
\item \emph{(Determinacy~\cite{Mar98})}. For all concurrent game structures and for all 
parity and multi-discounted objectives $\Phi$ we have 
$ 
 \sup_{\stra_1\in\bigstra_1}\inf_{\stra_2\in\bigstra_2}
  \Exp_s^{\stra_1,\stra_2}(\Phi) = \inf_{\stra_2\in\bigstra_2} \sup_{\stra_1\in\bigstra_1}
  \Exp_s^{\stra_1,\stra_2}(\Phi). 
$

\item \emph{(Memory complexity).} For all concurrent game structures and for all
multi-discounted objectives $\Phi$, randomized memoryless optimal strategies exist~\cite{Sha53}.
For all turn-based stochastic game structures and for all multi-discounted objectives
$\Phi$, pure memoryless optimal strategies exist~\cite{FV97}.
For all turn-based stochastic game strucutures and for all parity objectives $\Phi$, 
pure memoryless optimal strategies exist~\cite{CJH04,Zie04}.
In general optimal strategies need not exist in concurrent games with parity objectives,
and $\vare$-optimal strategies, for $\vare>0$, need both randomization and infinite
memory in general~\cite{dAH00}.
\end{enumerate}
\end{theorem}

The results of~\cite{dAHM03} established that the value of concurrent games with
certain special multi-discounted objectives can be characterized as valuations of
quantitaive discounted $\mu$-calculus formula.
In the limit, the value function of the discounted $\mu$-calculus formula characterizes
the value function of concurrent games with parity objectives.
An elegant interpretation of the result was given in~\cite{GZ05}, and from the 
interpretation we obtain the following theorem.

\begin{theorem}[\cite{dAHM03,GZ05}]\label{thrm_lit2}
Let $G$ be a concurrent game structure with a parity objective $\Phi$ defined by 
a priority function $p$. 
Let $r$ be a reward function that assigns reward~1 to even priority states and
reward~0 to odd priority states. 
Then there exists an order $s_1 s_2 \ldots s_n$ on the states (where $S=\set{s_1,s_2,\ldots,s_n}$) 
dependent only on the priority function $p$ such that 
$
\Val(G,\Phi) = 
\lim_{\disc(s_1) \to 1} \lim_{\disc(s_2) \to 1} \ldots \lim_{\disc(s_n) \to 1}
\Val(G,\MDT(\vdisc,r)); 
$
in other words, if we consider the value function $\Val(G,\MDT(\vdisc,r))$  
with the multi-discounted objective and take the limit of the discount factors to~1
in the order of the states we obtain the value function for the parity objective.

\end{theorem}

We now present notions related to \emph{structure equivalent} game 
structures and distances.

\smallskip\noindent{\bf Structure equivalent game structures.} 
Given two game structures $\game_1 = \langle S, \moves,\mov_1, \mov_2, \trans_1 \rangle$ and  
$\game_2 = \langle S, \moves,\mov_1, \mov_2, \trans_2 \rangle$ on the same 
state and action space, with different transition function, we say 
that $\game_1$ and $\game_2$ are \emph{structure equivalent} (denoted $\game_1 \equiv \game_2$) if 
for all $s \in S$ and all $a_1 \in \mov_1(s)$ and $a_2 \in \mov_2(s)$ 
we have $\supp(\trans_1(s,a_1,a_2))= \supp(\trans_2(s,a_1,a_2))$.
Similarly, two Markov chains $G_1=(S,\trans_1)$ and $G_2=(S,\trans_2)$  
are structurally equivalent (denoted $G_1\equiv G_2$) 
if for all $s\in S$ we have $\supp(\trans_1(s))=\supp(\trans_2(s))$.
For a game structure $G$ (resp. Markov chain $G$), we denote by 
$\equivclass{G}$ the set of all game structures (resp. Markov chains) that
are structurally equivalent to $G$.

\smallskip\noindent{\bf Ratio and absolute distances.} Given two game 
structures $\game_1=\langle S, \moves,\mov_1, \mov_2, \trans_1 \rangle$ and  
$\game_2 = \langle S, \moves,\mov_1, \mov_2, \trans_2 \rangle$, 
the \emph{absolute distance} of the game structures is maximum absolute 
difference in the transition probabilities.
Formally, 
$\dista(G_1,G_2)=\max_{s,t \in S,a \in \mov_1(s),b\in \mov_2(s)} |\trans_1(s,a,b)(t) -\trans_2(s,a,b)(t)|$.
The absolute distance for two Markov chains $G_1=(S,\trans_1)$ and 
$G_2=(S,\trans_2)$ is $\dista(G_1,G_2)=\max_{s,t\in S} |\trans_1(s)(t)-\trans_2(s)(t)|$.
We now define the ratio distance between two structurally equivalent 
game structures and Markov chains. 
Let $G_1$ and $G_2$ be two structurally equivalent game structures.
The \emph{ratio} distance is defined on the ratio of the transition 
probabilities.
Formally,
\[
\begin{array}{rcl}
\dist(G_1,G_2) & = & 
\displaystyle  \max\Bigg\{ 
\frac{\trans_1(s,a,b)(t)}{\trans_2(s,a,b)(t)},
\frac{\trans_2(s,a,b)(t)}{\trans_1(s,a,b)(t)} \mid s\in S, a\in \mov_1(s),b \in \mov_2(s), \\
 & & \qquad \quad t\in \supp(\trans_1(s,a,b))=\supp(\trans_2(s,a,b))\Bigg\} -1
\end{array}
\]
The ratio distance between two structurally equivalent Markov chains
$G_1$ and $G_2$ is $\max\big\{ \frac{\trans_1(s)(t)}{\trans_2(s)(t)},
\frac{\trans_2(s)(t)}{\trans_1(s)(t)} \mid s\in S, t\in \supp(\trans_1(s))=\supp(\trans_2(s))\big\} -1 $.

\smallskip\noindent{\bf Remarks about the distance functions.}
We first remark that the ratio distance is not necessarily a metric.
Consider the Markov chain with state space $S=\set{s,s'}$ and 
let $\vare\in (0,1/7)$. 
For $k=1,2,5$ consider the transition functions $\trans_k$ such that 
$\trans_k(t)(s)=1-\trans_k(t)(s')=k\cdot \vare$, for all $t\in S$. 
Let $G_k$ be the Markov chain with transition function $\trans_k$.
Then we have $\dist(G_1,G_2)=1, \dist(G_2,G_5)=\frac{3}{2}$ and 
$\dist(G_1,G_5)=4$, and hence $\dist(G_1,G_2) +\dist(G_2,G_5) 
< \dist(G_1,G_5)$. The above example is from~\cite{Sol03}.
Also note that $\dist$ is only defined for structurally equivalent
game structures, and without the assumption $\dist$ is $\infty$.
We also remark that the absolute distance that measures the difference 
in the transition probabilities is the most intuitive measure for the 
difference of two game structures.

\begin{proposition}\label{prop_dist}
Let $G_1$ be a game structure (resp. Markov chain) such that the minimum 
positive transition probability is $\eta>0$. 
For all game structures (resp. Markov chains) $G_2 \in \equivclass{G_1}$ 
we have 
$\dist(G_1,G_2) \leq \frac{\dista(G_1,G_2)}{\eta}$.
\end{proposition}

\smallskip\noindent{\bf Notation for fixing strategies.} 
Given a concurrent game structure $G= \langle S, \moves,\mov_1, \mov_2, \trans \rangle$,
let $\stra_1$ be a randomized memoryless strategy. 
Fixing the strategy $\stra_1$ in $G$ we obtain a player-2 MDP, denoted as
$G\restr \stra_1$, defined as follows: 
(1) the state space is $S$; (2) for all $s \in S$ we have $\mov_1(s)=\set{\bot}$ (hence 
it is a player-2 MDP); (3) the new transition function $\trans_{\stra_1}$ is defined as
follows: for all $s \in S$ and all $b\in \mov_2(s)$ we have 
$\trans_{\stra_1}(s,\bot,b)(t)=\sum_{a\in \mov_1(s)} \stra_1(s)(a)\cdot \trans(s,a,b)(t)$.
Similarly if we fix a randomized memoryless strategy $\stra_1$ in an MDP $G$ we obtain a
Markov chain, denoted as $G\restr \stra_1$.
The following proposition is straightforward to verify from the definitions.

\begin{proposition}\label{prop_stra_fix}
Let $G_1$ and $G_2$ be two concurrent game structures (resp. MDPs) that are 
structurally equivalent.
Let $\stra_1$ be a randomized memoryless strategy.
Then $\dista(G_1 \restr\stra_1, G_2 \restr \stra_1) \leq \dista(G_1,G_2)$ 
and $\dist(G_1 \restr\stra_1, G_2 \restr \stra_1) \leq \dist(G_1,G_2)$. 
\end{proposition}

\section{Markov Chains with Multi-discounted and Parity Objectives}\label{sec:mc}
In this section we consider Markov chains with multi-discounted
and parity objectives.
We present a bound on the difference of value functions of two
structurally equivalent Markov chains that is dependent on the 
distance between the Markov chains and is \emph{independent} of
the discount factors.
The result for parity objectives is then a consequence of our result
for multi-discounted objectives and Theorem~\ref{thrm_lit2}.
Our result crucially depends on a result of Friedlin and Wentzell
for Markov chains and we present this result below, and
then use it to obtain the main result of the section.

\smallskip\noindent{\bf Result of Friedlin and Wentzell.}
Let $(S,\trans)$ be a Markov chain and let $s_0$ be the initial state.
Let $C \subset S$ be a proper subset of $S$ and let us denote by 
$\exit_C=\inf\set{n\in \Nats\mid \theta_n \not\in C}$ the first
hitting time to the set $S\setminus C$ of states (or the first exit time
from set $C$) (recall that $\theta_n$ 
is the random variable to denote the $n$-th state of a path).
Let $\calf(C,S)=\set{f: C \mapsto S}$ denote the set of all functions from 
$C$ to $S$.
For every $f \in \calf(C,S)$ we define a directed graph $G_f=(S,E_f)$ 
where $(s,t)\in E_f$ iff $f(s)=t$.
Let $\alpha_f=1$ if the directed graph $G_f$ has no directed cycles
(i.e., $G_f$ is a directed acyclic graph); and $\alpha_f=0$ otherwise.
Observe that since $f$ is a function, for every $s\in C$ there is exactly
one path that starts at $s$.
For every $s \in C$ and every $t \in S$, let  
$\beta_f(s,t)=1$ if the directed path that leaves $s$ in $G_f$ reaches $t$,
otherwise $\beta_f(s,t)=0$.
We now state a result that can be obtained as a special case of the
result from Friedlin and Wentzell~\cite{FW84}.
Below we use the formulation of the result as presented in~\cite{Sol03}
(Lemma~2 of~\cite{Sol03}).

\begin{theorem}[Friedlin-Wentzell result~\cite{FW84}]\label{thrm-fw}
Let $(S,\trans)$ be a Markov chain, and let $C\subset S$ be a proper subset 
of $S$ such that $\Prb_{s}(\exit_C < \infty) >0$ for every $s \in C$ (i.e., 
from all $s \in C$ with positive probability the first hitting time to
the complement set is finite).
Then for every initial state $s_1\in C$ and for every $t \not\in C$ we have 
\begin{eqnarray}\label{eq1}
\Prb_{s_1}(\theta_{\exit_C}=t)= 
\frac{\sum_{f\in \calf(C,S)} (\beta_f(s_1,t) \cdot \prod_{s\in C} \trans(s)(f(s)) )} 
{\sum_{f\in \calf(C,S)} (\alpha_f \cdot \prod_{s\in C} \trans(s)(f(s)))},
\end{eqnarray}
in other words, the probability that the exit state is $t$ when the starting state
is $s_1$ is given by the expression on the right hand side (very informally the right
hand side is the normalized polynomial expression for exit probabilities).
\end{theorem}
\begin{comment}
We present an argument that the assumption that for all $s \in C$ we have 
$\Prb_{s}(\exit_C < \infty) >0$ implies that the denominator of Equation (\ref{eq1})
is positive (also see~\cite{Milman,FW84,Sol03}).
Since all terms in the summation of the denominator is non-negative, we show 
a witness function $f \in \calf(C,S)$ such that $\alpha_f=1$ and 
$\prod_{s\in C} \trans(s)(f(s)) >0$.
Let $s \in C$, and since $\Prb_{s}(\exit_C < \infty) >0$, it follows that there 
exists $\ell >1$ and a sequence of states $s_1 s_2 \ldots s_\ell$ with $s_1=s$ such 
that $s_2,\ldots,s_{\ell-1} \in C$, $s_\ell \in (S\setminus C)$ and for all 
$i=1,2,\ldots, \ell-1$ we have $\trans(s_i)(s_{i+1})>0$.
Let us denote by $\ell_s$ the length of the shortest such sequence. 
We have the following two cases: 
(1) $\ell_s=2$, i.e., there exists $t \in (S\setminus C)$ and $\trans(s)(t)>0$;
or (2) $\ell_s>2$, and then there exists $t \in C$ with $\trans(s)(t)>0$ and 
$\ell_s=\ell_t+1$.
We define the witness $f$ as follows: (1) if $\ell_s=2$, then $f(s)=t$, where $t$ is any
state in $S\setminus C$ with $\trans(s)(t)>0$; 
(2) if $\ell_s>2$, then $f(s)=t$, where $t \in C$ is a state in $C$ such that 
$\trans(s)(t)>0$ and $\ell_s=\ell_t+1$.
Since $s \in S$ is chosen arbitrarily, $f$ is a function from $C$ to $S$, and 
by construction we have $\prod_{s\in C} \trans(s)(f(s)) >0$.
Since for every $s \in C$, if $f(s) \in C$, then $\ell_{f(s)}+1=\ell_s$, it follows
that the directed graph induced by $f$ has no cycles and hence $\alpha_f=1$.
\end{comment}

\noindent{\bf Value function difference for Markov chains.}
We will use the result of Theorem~\ref{thrm-fw} to obtain 
bounds on the value functions of Markov chains. We start with the
notion of mean-discounted time.

\smallskip\noindent{\bf Mean-discounted time.}
Given a Markov chain $(S,\trans)$ and a discount vector $\vdisc$, we define for every 
state $s \in S$, the \emph{mean-discounted time} the process is in the state $s$.
We first define the mean-discounted time function $\MDT(\vdisc,s): \Paths \mapsto \reals$ that
maps every path to the mean-discounted time that the state $s$ is visited, and 
the function is formally defined as follows: for a path $\pat=s_0 s_1 s_2 \ldots$ we have  
\[
\MDT(\vdisc,s)(\pat)= \frac{\sum_{j=0}^\infty (\prod_{i=0}^j \disc(s_i) ) \cdot \indi_{s_j=s} }
{\sum_{j=0}^\infty (\prod_{i=0}^j \disc(s_i) ) };
\]
where $\indi_{s_j=s}$ is the indicator function.
The expected mean-discounted time function for a Markov chain $G$ with transition function 
$\trans$ is defined as follows:
$\MT(s_1,s,G,\vdisc)=\Exp_{s_1}[\MDT(\vdisc,s)]$, i.e., it is the expected mean-discounted time for 
$s$ when the starting state is $s_1$, where the expectation measure is defined by the
Markov chain with transition function $\trans$.
We now present a lemma that shows the value function for multi-discounted Markov chains 
can be expressed as ratio of two polynomials (the result is obtained as a simple 
extension of a result of Solan~\cite{Sol03}).

\begin{lemma}\label{lemm:ratio}
For Markov chains defined on state space $S$, for all initial states $s_0$, 
for all states $s$, for all discount vectors $\vdisc$, there exists two polynomials 
$g_1(\cdot)$ and $g_2(\cdot)$ in $|S|^2$ variables $x_{t,t'}$, where $t,t'\in S$ such that
the following conditions hold:
\begin{enumerate}
\item the polynomials have degree at most $|S|$ with non-negative coefficients; and
\item for all transition functions $\trans$ over $S$ we have 
$\MT(s_0,s,G,\vdisc)= \frac{g_1(\trans)}{g_2(\trans)}$, where $G=(S,\trans)$, 
$g_1(\trans)$ and $g_2(\trans)$ denote the values of the function $g_1$ and $g_2$ such that 
all the variables $x_{t,t'}$ is instantiated with values $\trans(t)(t')$ as given by the transition 
function $\trans$.
\end{enumerate} 
\end{lemma}
\begin{proof} {\em (Sketch).} We present a sketch of the proof (details in appendix).
Fix a discount vector $\vdisc$. We construct a Markov chain $\ov{G}=(\ov{S},\ov{\trans})$ as 
follows: $\ov{S}=S \cup S_1$, where $S_1$ is a copy of states of $S$ (and for a state $s \in S$ 
we denote its corresponding copy as $s_1$); and the transition function $\ov{\trans}$ 
is defined below
\begin{enumerate}
\item $\ov{\trans}(s_1)(s_1)=1$ for all $s_1 \in S_1$ (i.e., all copy states are absorbing);
\item for $s \in S$ we have 
\[
\ov{\trans}(s)(t) =
\begin{cases}
(1-\disc(s)) & t=s_1; \\
\disc(s) \cdot \trans(s)(t) & t\in S; \\
0 & t\in S_1\setminus{s_1};
\end{cases}
\]
i.e., it goes to the copy with probability $(1-\disc(s))$, it follows
the transition $\trans$ in the original copy with probabilities multiplied 
by $\disc(s)$.
\end{enumerate}
We first show that for all $s_0$ and $s$ we have
$\MT(s_0,s,G,\vdisc)=
\Prb_{s_0}^{\ov{\trans}}(\theta_{\exit_S}=s_1)$;
i.e., the expected mean-discounted time in $s$ when the original Markov chain 
starts in $s_0$ is the probability in the Markov chain $(\ov{S},\ov{\trans})$ 
that the first hitting state out of $S$ is the copy $s_1$ of the state $s$.
The claim is easy to verify as both $(\MT(s_0,s,G,\vdisc))_{s_0\in S}$ and 
$(\Prb_{s_0}^{\ov{\trans}}(\theta_{\exit_S}=s_1))_{s_0\in S}$ are the unique solution 
of the following system of linear equations: for all $t \in S$ we have
$
y_{t}= (1-\disc(t)) \cdot \indi_{t=s} + \sum_{z\in S} \disc(t)\cdot\trans(t)(z) \cdot y_z.
$

We now claim that 
$\Prb_{s_0}^{\ov{\trans}}(\exit_{S} < \infty) >0$ for all $s_0 \in S$.
This follows since for all $s \in S$ we have $\ov{\trans}(s)(s_1) =(1-\disc(s))>0$ 
and since $s_1 \not\in S$ we have 
$\Prb_{s_0}^{\ov{\trans}}(\exit_{S} =2)= (1-\disc(s_0)) >0$.
Now we observe that we can apply Theorem~\ref{thrm-fw} on the Markov chain 
$\ov{G}=(\ov{S},\ov{\trans})$  with $S$ as the set $C$ of states of Theorem~\ref{thrm-fw}, 
and obtain the result.
Indeed the terms $\alpha_f$ and $\beta_f(s,t)$ are independent of 
$\trans$, and the two products of Equation (\ref{eq1}) each contains at most $|S|$
terms of the form $\ov{\trans}(s)(t)$ for $s,t \in \ov{S}$.
Thus the desired result follows.
\qed
\end{proof}

\begin{lemma}\label{lemm:poly}
Let $h(x_1,x_2,\ldots,x_k)$ be a polynomial function with non-negative
coefficients of degree at most $n$.
Let $\vare>0$ and $y,y' \in \reals^k$ be two non-negative vectors such that
for all $i=1,2,\ldots,k$ we have $\frac{1}{1+\vare} \leq \frac{y_i}{y'_i} 
\leq 1+\vare$.
Then we have 
$(1+\vare)^{-n} \leq \frac{h(y)}{h(y')} \leq (1+\vare)^n$.
\end{lemma}

\begin{lemma}\label{lemm_mc_multi}
Let $G_1=(S,\trans)$ and $G_2=(S,\trans')$ be two structurally  equivalent Markov
chains. 
For all non-negative reward functions $r:S \mapsto \reals$ such that the reward 
function is bounded by~1, for all discount vectors $\vdisc$, 
for all $s \in S$ we have 
$|\Val(G_1,\MDT(\vdisc,r))(s) - \Val(G_2,\MDT(\vdisc,r))(s)| \leq 
(1+ \dist(G_1,G_2))^{2\cdot|S|}-1$;
i.e., the absolute difference of the value functions for the multi-discounted objective is 
bounded by $(1+ \dist(G_1,G_2))^{2\cdot|S|}-1$.
\end{lemma}

\noindent The proof of Lemma~\ref{lemm_mc_multi} uses Lemma~\ref{lemm:ratio} and Lemma~\ref{lemm:poly} 
and is presented in the appendix.

\begin{theorem}\label{thrm_mc_multi}
Let $G_1=(S,\trans)$ and $G_2=(S,\trans')$ be two structurally  equivalent Markov
chains. 
Let $\eta$ be the minimum positive transition probability in $G_1$.
The following assertions hold:
\begin{enumerate}
\item 
For all non-negative reward functions $r:S \mapsto \reals$ such that the reward 
function is bounded by~1, for all discount vectors $\vdisc$, 
for all $s \in S$ we have 
\[
\begin{array}{rcl}
|\Val(G_1,\MDT(\vdisc,r))(s) - \Val(G_2,\MDT(\vdisc,r))(s)| & \leq & 
(1+ \vare_R)^{2\cdot|S|}-1 
\\[2ex] & \leq &\displaystyle
(1+\vare_A)^{2\cdot |S|}-1
\end{array}
\]

\item For all parity objectives $\Phi$ and  
for all $s \in S$ we have 
\[
\begin{array}{rcl}
|\Val(G_1,\Phi)(s) - \Val(G_2,\Phi)(s)| & \leq & 
(1+ \vare_R)^{2\cdot|S|}-1 
\leq  
(1+\vare_A)^{2\cdot |S|}-1
\end{array}
\]

\end{enumerate}
where $\vare_R=\dist(G_1,G_2)$ and 
$\vare_A=\frac{\dista(G_1,G_2)}{\eta}$.
\end{theorem}
\begin{proof}
The first part follows from Lemma~\ref{lemm_mc_multi} and 
Proposition~\ref{prop_dist}.
The second part follows from part~1, the fact the value 
function for parity objectives is obtained as the limit 
of multi-discounted objectives (Theorem~\ref{thrm_lit2}), and 
the fact the bound for part~1 is independent of the 
discount factors (hence independent of taking the limit).
\qed
\end{proof}

\smallskip\noindent{\bf Remark on structural assumption in the proof.}
The result of the previous theorem depends on the structural equivalence 
assumption in two crucial ways. They are as follows:
(1)~Proposition~\ref{prop_dist} that establishes the relation of 
$\dist$ and $\dista$ only holds with the assumption of structural 
equivalence; and 
(2)~without the structural equivalence assumption $\dist$ is $\infty$, and 
hence without the assumption the bound of
the previous theorem is $\infty$, which is a trivial bound.
We will later show (in Example~\ref{examp1}) that the structural equivalence
assumption is necessary.

\section{Value Continuity for Parity Objectives}\label{sec:games}
In this section we show two results: first we show robustness of strategies 
and present quantitative bounds on value functions for turn-based stochastic 
games and then we show continuity for concurrent parity games.

\subsection{Bounds for structurally equivalent turn-based stochastic parity games}\label{sec:tb}
In this section we present quantitative bounds for robustness of optimal strategies 
in structurally equivalent turn-based stochastic games.
For every $\vare>0$, we present a bound $\beta>0$, such 
that if the distance of the structurally equivalent turn-based stochastic games differs 
by at most $\beta$, then any pure memoryless optimal strategy in one game is $\vare$-optimal
in the other. 
The result is first shown for MDPs and then extended to turn-based stochastic games (both 
proofs are in the appendix).

\begin{theorem}\label{thrm_tb_stochastic}
Let $G_1$ be a turn-based stochastic game such that the minimum positive transition 
probability is $\eta>0$. 
The following assertions hold:
\begin{enumerate}
\item For all turn-based stochastic games $G_2 \in \equivclass{G_1}$, 
for all parity objectives $\Phi$ and for all $s \in S$ we have 
\[
\begin{array}{rcl}
|\Val(G_1,\Phi)(s) - \Val(G_2,\Phi)(s)| 
& \leq & \displaystyle
(1+ \dist(G_1,G_2))^{2\cdot|S|}-1 \\[2ex]
& \leq & \displaystyle
\bigg(1+\frac{\dista(G_1,G_2)}{\eta}\bigg)^{2\cdot |S|}-1
\end{array}
\]

\item For $\vare>0$, let 
$\beta \leq \frac{\eta}{2} \cdot\big((1+\frac{\vare}{2})^{\frac{1}{2\cdot|S|}}-1)$.
For all $G_2 \in \equivclass{G_1}$ such that $\dista(G_1,G_2) \leq \beta$, for all 
parity objectives $\Phi$, every pure memoryless
optimal strategy $\stra_1$ in $G_1$ is an $\vare$-optimal strategy in $G_2$. 
\end{enumerate}
\end{theorem}

\subsection{Value continuity for concurrent parity games}\label{sec:conc}

In this section we show value continuity for structurally equivalent 
concurrent parity games, and show with an example on Markov chains that the 
continuity property breaks without the structural equivalence assumption.
Finally with an example on Markov chains we show the our quantitative bounds
are asymptotically optimal for small distance values.
We start with a lemma for MDPs.

\begin{lemma}\label{lemm_mdp_multi}
Let $G_1$ and $G_2$ be two structurally  equivalent MDPs. 
Let $\eta$ be the minimum positive transition probability in $G_1$.
For all non-negative reward functions $r:S \mapsto \reals$ such that the reward 
function is bounded by~1, for all discount vectors $\vdisc$, 
for all $s \in S$ we have 
\[
\begin{array}{rcl}
|\Val(G_1,\MDT(\vdisc,r))(s) - \Val(G_2,\MDT(\vdisc,r))(s)| & \leq & 
(1+ \dist(G_1,G_2))^{2\cdot|S|}-1 \\[2ex]
& \leq &
\displaystyle 
\bigg(1+\frac{\dista(G_1,G_2)}{\eta}\bigg)^{2\cdot |S|}-1
\end{array}
\]
\end{lemma}

The main idea of the proof of the above lemma is to fix a pure memoryless 
optimal strategy and then use the results for Markov chains. 
Using the same proof idea, along with randomized memoryless optimal
strategies for concurrent game structures and the above lemma, we obtain 
the following lemma (the result is identical to the previous lemma, but
for concurrent game structures instead of MDPs).

\begin{lemma}\label{lemm_conc_multi}
Let $G_1$ and $G_2$ be two structurally  equivalent concurrent game structures. 
Let $\eta$ be the minimum positive transition probability in $G_1$.
For all non-negative reward functions $r:S \mapsto \reals$ such that the reward 
function is bounded by~1, for all discount vectors $\vdisc$, 
for all $s \in S$ we have 
\[
\begin{array}{rcl}
|\Val(G_1,\MDT(\vdisc,r))(s) - \Val(G_2,\MDT(\vdisc,r))(s)| & \leq & 
(1+ \dist(G_1,G_2))^{2\cdot|S|}-1 \\[2ex]
& \leq &
\displaystyle 
\bigg(1+\frac{\dista(G_1,G_2)}{\eta}\bigg)^{2\cdot |S|}-1
\end{array}
\]
\end{lemma}

We now present the main theorem that depends on Lemma~\ref{lemm_conc_multi}.

\begin{theorem}\label{thrm_conc_multi}
Let $G_1$ and $G_2$ be two structurally  equivalent concurrent game structures. 
Let $\eta$ be the minimum positive transition probability in $G_1$.
For all parity objectives $\Phi$ and
for all $s \in S$ we have 
\[
\begin{array}{rcl}
|\Val(G_1,\Phi)(s) - \Val(G_2,\Phi)(s)| & \leq & 
(1+ \dist(G_1,G_2))^{2\cdot|S|}-1 \\[2ex] & \leq &
\displaystyle 
\bigg(1+\frac{\dista(G_1,G_2)}{\eta}\bigg)^{2\cdot |S|}-1
\end{array}
\]
\end{theorem}
\begin{proof}
The result follows from Theorem~\ref{thrm_lit2}, Lemma~\ref{lemm_conc_multi}
and the fact that the bound of Lemma~\ref{lemm_conc_multi} are independent 
of the discount factors and hence independent of taking the limits.
\qed
\end{proof}

In the following theorem we show that for structurally equivalent 
game structures, for all parity objectives, the value function 
is continuous in the absolute distance between the game structures.
We have already remarked (after Theorem~\ref{thrm_mc_multi}) that the structural
equivalence assumption is required in our proofs, and we show in Example~\ref{examp1}
that this assumption is necessary.

\begin{theorem}\label{thrm_val_con}
For all concurrent game structures $G_1$, for all parity objectives $\Phi$ 
\[
\lim_{\vare \to 0} \ \ \sup_{G_2 \in \equivclass{G_1}, \dista(G_1,G_2) \leq \vare} \ \ \sup_{s\in S} 
|\Val(G_1,\Phi)(s) - \Val(G_2,\Phi)(s)| =0.
\]
\end{theorem}
\begin{proof}
Let $\eta>0$ be the minimum positive transition probability in $G_1$.
By Theorem~\ref{thrm_conc_multi} we have
\[
\lim_{\vare \to 0}\ \sup_{G_2 \in \equivclass{G_1}, \dista(G_1,G_2) \leq \vare} \  \sup_{s\in S} 
|\Val(G_1,\Phi)(s) - \Val(G_2,\Phi)(s)| \leq
\lim_{\vare \to 0} 
\bigg(1+\frac{\vare}{\eta}\bigg)^{2\cdot |S|}-1
\]
The above limit equals to~0, and the desired result follows.
\qed
\end{proof}

\begin{example}[Structurally equivalence assumption necessary]\label{examp1}
In this example we show that in Theorem~\ref{thrm_val_con} the structural equivalence 
assumption is necessary, and thereby show that the result is tight.
We show an Markov chain $G_1$ and a family of Markov chains $G_2^\vare$, for $\vare>0$, 
such that $\dista(G_1,G_2^\vare) \leq \vare$ (but $G_1$ is not structurally equivalent
to $G_2^\vare$) with a parity objective $\Phi$ and we have
$\lim_{\vare \to 0} \sup_{s \in S} |\Val(G_1,\Phi)(s)-\Val(G_2^\vare,\Phi)(s)|=1$.
The Markov chains $G_1$ and $G_2^\vare$ are defined over the state
space $\set{s_0,s_1}$, and in $G_1$ both states have self-loops with 
probability~1, and in $G_2^\vare$ the self-loop at $s_0$ has probability
$1-\vare$ and the transition probability from $s_0$ to $s_1$ is
$\vare$ (see Fig~\ref{figure:buchi-lim} in appendix). 
Clearly, $\dista(G_1,G_2^\vare) =\vare$. 
The parity objective $\Phi$ requires to visit the state $s_1$ infinitely
often (i.e., assign priority~2 to $s_1$ and priority~1 to $s_0$).
Then we have $\Val(G_1,\Phi)(s_0)=0$ as the state $s_0$ is never left,
whereas in $G_2^\vare$ the state $s_1$ is the only closed recurrent set 
of the Markov chain and hence reached with probability~1 from $s_0$.
Hence $\Val(G_2^\vare,\Phi)(s_0)=1$.
It follows that $
\lim_{\vare \to 0} \sup_{s \in S} |\Val(G_1,\Phi)(s)-\Val(G_2^\vare,\Phi)(s)|=1$.
\qed
\end{example}

\begin{example}[Asymptotically tight bound for small distances] We now show that 
our quantitative bound for the value function difference is asymptotically optimal 
for small distances. 
Let us denote the absolute distance as $\vare$, and the quantitative bound we 
obtain in Theorem~\ref{thrm_conc_multi} is 
$(1+\frac{\vare}{\eta})^{2\cdot|S|}-1$, and if $\vare$ is small, then 
we obtain the following approximate bound 
\[
\bigg(1+\frac{\vare}{\eta}\bigg)^{2\cdot |S|}-1 \approx 1 + 2\cdot |S| \cdot \frac{\vare}{\eta} -1 
= 2\cdot |S| \cdot \frac{\vare}{\eta}. 
\]
We now illustrate with an example (on structurally equivalent Markov chains) where the difference in the value 
function is $O(|S|\cdot \vare)$, for small $\vare$.
Consider the Markov chain defined on state space $S=\set{s_0,s_1,\ldots, s_{2n-1},s_{2n}}$
as follows: 
states $s_0$ and $s_{2n}$ are absorbing (states with self-loops of 
probability~1) and for a state $1\leq i \leq 2n-1$ we have 
$\trans(s_i)(s_{i-1})=\frac{1}{2}+\vare$; and $\trans(s_i)(s_{i+1})=\frac{1}{2}-\vare$;
i.e., we have a Markov chain defined on a line from $0$ to $2n$ (with $0$ and $2n$ absorbing
states) and the chain moves towards $0$ with probability $\frac{1}{2}+\vare$ and towards
$2n$ with probability $\frac{1}{2}-\vare$ (see Fig~\ref{figure:asym} with complete details in appendix). 
Our goal is to estimate the probability to reach the state $s_0$, and let $v_i$ denote 
the probability to reach $s_0$ from the starting state $s_i$.
We show (details in appendix) that if $\vare=0$, then $v_n=\frac{1}{2}$ and for $0<\vare<\frac{1}{2}$, 
such that $\vare$ is close to~0, we have $v_n=\frac{1}{2}+n\cdot \vare$.
Observe that the Markov chain obtained for $\vare=0$ and $\frac{1}{2} > 
\vare >0$ are structurally equivalent. Thus the desired result follows.
\qed
\end{example}

\section{Conclusion}
In this work we studied the robustness and continuity property of 
concurrent and turn-based stochastic parity games with respect to 
small imprecision in the transition probabilities.
We presented (i)~quantitative bounds on difference of the value functions 
and proved value continuity for concurrent parity games under the
structural equivalence assumption, and (ii)~showed robustness of all 
pure memoryless optimal strategies for structurally equivalent 
turn-based stochastic parity games.
We also showed that the structural equivalence assumption is
necessary and that our quantitative bounds are asymptotically
optimal for small imprecision. 
We believe our results will find applications in robustness analysis 
of various other classes of stochastic games.

\clearpage

\section*{Appendix}

\section{Missing proofs of Section~2}
\begin{proof} {\em (of Proposition~\ref{prop_dist}).}
Consider $s \in S$, $a\in \mov_1(s),b\in \mov_2(s)$, and $t \in \supp(\trans_1(s,a,b))=\supp(\trans_2(s,a,b))$.
Then we have the following two inequalities: we consider 
$\frac{\trans_2(s,a,b)(t)}{\trans_1(s,a,b)(t)}$, and the argument 
for 
$\frac{\trans_1(s,a,b)(t)}{\trans_2(s,a,b)(t)}$ is symmetric.
We consider $\frac{\trans_2(s,a,b)(t)}{\trans_1(s,a,b)(t)}$ and 
if $\trans_2(s,a,b)(t) \leq \trans_1(s,a,b)(t)$, then 
$\frac{\trans_2(s,a,b)(t)}{\trans_1(s,a,b)(t)} \leq 1$, and 
otherwise we have the following inequality:
\[
\frac{\trans_2(s,a,b)(t)}{\trans_1(s,a,b)(t)} 
\leq 
\frac{\trans_1(s,a,b)(t)+\dista(G_1,G_2)}{\trans_1(s,a,b)(t)} 
= 1 +\frac{\dista(G_1,G_2)}{\trans_1(s,a,b)(t)} \leq 1+\frac{\dista(G_1,G_2)}{\eta}. 
\]
It follows that in both cases we have 
$\frac{\trans_2(s,a,b)(t)}{\trans_1(s,a,b)(t)} -1 \leq \frac{\dista(G_1,G_2)}{\eta}$.
The desired result follows from the above inequalities.
\qed
\end{proof}

\section{Missing proofs of Section~3}
We now  present the proof of Lemma~\ref{lemm:ratio} which is obtained as a 
simple extension of a result of Solan~\cite{Sol03}.

\begin{proof} {\em (of Lemma~\ref{lemm:ratio}).}
Fix a discount vector $\vdisc$. We construct a Markov chain $\ov{G}=(\ov{S},\ov{\trans})$ as 
follows: $\ov{S}=S \cup S_1$, where $S_1$ is a copy of states of $S$ (and for a state $s \in S$ 
we denote its corresponding copy as $s_1$); and the transition function $\ov{\trans}$ 
is defined below
\begin{enumerate}
\item $\ov{\trans}(s_1)(s_1)=1$ for all $s_1 \in S_1$ (i.e., all copy states are absorbing);
\item for $s \in S$ we have 
\[
\ov{\trans}(s)(t) =
\begin{cases}
(1-\disc(s)) & t=s_1; \\
\disc(s) \cdot \trans(s)(t) & t\in S; \\
0 & t\in S_1\setminus{s_1};
\end{cases}
\]
i.e., it goes to the copy with probability $(1-\disc(s))$, it follows
the transition $\trans$ in the original copy with probabilities multiplied 
by $\disc(s)$.
\end{enumerate}
We first show that for all $s_0$ and $s$ we have
\[
\MT(s_0,s,G,\vdisc)=
\Prb_{s_0}^{\ov{\trans}}(\theta_{\exit_S}=s_1);
\]
i.e., the expected mean-discounted time in $s$ when the original Markov chain 
starts in $s_0$ is the probability in the Markov chain $(\ov{S},\ov{\trans})$ 
that the first hitting state out of $S$ is the copy $s_1$ of the state $s$.
The claim is easy to verify as both $(\MT(s_0,s,G,\vdisc))_{s_0\in S}$ and 
$(\Prb_{s_0}^{\ov{\trans}}(\theta_{\exit_S}=s_1))_{s_0\in S}$ are the solutions 
of the following system of linear equations
\[
y_{t}= (1-\disc(t)) \cdot \indi_{t=s} + \sum_{z\in S} \disc(t)\cdot\trans(t)(z) \cdot y_z \quad \forall t \in S.
\]
The fact that $(\MT(s_0,s,G,\vdisc))_{s_0\in S}$ is the solution of the above 
equation follows from the results of discounted reward Markov chains 
(detailed proofs with uniform discount factor for MDPs is available in~\cite{FV97} 
(e.g., equation 2.15 of~\cite{FV97}),
and specialization to Markov chains and generalization to discount factor attached
to every state is straightforward).
The fact that $(\Prb_{s_0}^{\ov{\trans}}(\theta_{\exit_S}=s_1))_{s_0\in S}$ 
is the solution of the above equation follows from the results of characterization 
of hitting time for transient Markov chains (see~\cite{Derman} for details).
Also the above system of linear equations has a unique solution.
The uniqueness of the solution follows from the fact that this is a 
contraction mapping, and the proof is as follows:
let $(y^1_z)_{z \in S}$ and $(y^2_z)_{z \in S}$ be two solutions of the system.
We chose $z^*\in S$ such that $z^*=\arg\max_{z\in S} |y^1_z -y^2_z|$, i.e., $z^*$ is 
a state that maximizes the difference of the two solutions.
Let $\eta=|y^1_{z^*} - y^2_{z^*}|$. 
As $y^1$ and $y^2$ are solutions of the above system we have by the triangle 
inequality
\[
\begin{array}{rcl}
0 \leq \eta =|y^1_{z^*} - y^2_{z^*}| & \leq & 
\displaystyle \sum_{t \in S} \disc(t) \cdot |y^1_t-y^2_t| \\[1ex] 
& \leq & 
\displaystyle \eta \cdot \sum_{t\in S} \disc(t) \cdot \trans(s_0)(t) 
\leq \eta \cdot \max_{t\in S} \disc(t) \cdot \sum_{t \in S}\trans(s_0)(t).
\end{array}
\] 
Since $\sum_{t \in S}\trans(s_0)(t)=1$, it follows that 
$\eta \leq \eta \cdot\max_{t\in S} \disc(t)$.
Since $\max_{t\in S}\disc(t) < 1$ it follows that we must have $\eta=0$ and hence the 
two solutions must coincide.

We now claim that 
$\Prb_{s_0}^{\ov{\trans}}(\exit_{S} < \infty) >0$ for all $s_0 \in S$.
This follows since for all $s \in S$ we have $\ov{\trans}(s)(s_1) =(1-\disc(s))>0$ 
and since $s_1 \not\in S$ we have 
$\Prb_{s_0}^{\ov{\trans}}(\exit_{S} =2)= (1-\disc(s_0)) >0$.
Now we observe that we can apply Theorem~\ref{thrm-fw} on the Markov chain 
$\ov{G}=(\ov{S},\ov{\trans})$  with $S$ as the set $C$ of states of Theorem~\ref{thrm-fw}, 
and obtain the result.
Indeed the terms $\alpha_f$ and $\beta_f(s,t)$ are independent of 
$\trans$, and the two products of Equation (\ref{eq1}) each contains at most $|S|$
terms of the form $\ov{\trans}(s)(t)$ for $s,t \in \ov{S}$.
Thus the desired result follows.
\qed
\end{proof}

\begin{example}[Illustration of construction of Lemma~\ref{lemm:ratio}]
We now illustrate the construction of Lemma~\ref{lemm:ratio} with the aid of 
some examples.
Consider the Markov chain $G$ with states $s$ and $t$ such that 
$t$ is absorbing and the transition from $s$ to $t$ has 
probability~1, and let the discount factor be $1/3$ for all states.
The Markov chain $G$ along with $\ov{G}$ is shown in Fig.~\ref{fig:illus_1}.
If we start at $s$, the mean-discounted time at $t$ is given by 
\[
\frac{1/3^2 + 1/3^3 + \ldots }{ 1/3 + 1/3^2 + 1/3^3 +\ldots}
=\frac{1/9\cdot 3/2}{1/3 \cdot 3/2}= \frac{1}{3}.
\]
In the Markov chain $\ov{G}$, the probability to reach $t$ from $s$ is 
$1/3$, and once $t$ is reached the exit state is $t_1$ with probability~1.
Hence the probability to exit through state $t_1$ is also $1/3$.

\begin{figure*}[t]
\begin{center}
\begin{picture}(85,25)(0,0)
\node[Nmarks=n](n0)(0,18){$s$}
\node[Nmarks=n](n1)(20,18){$t$}
\node[Nmarks=n](n2)(70,18){$s$}
\node[Nmarks=n](n3)(90,18){$t$}
\node[Nmarks=n](n4)(70,0){$s_1$}
\node[Nmarks=n](n5)(90,0){$t_1$}
\drawloop(n1){$1$}
\drawloop(n3){$1/3$}
\drawloop[loopangle=180](n4){$1$}
\drawloop[loopangle=0](n5){$1$}
\drawedge[ELpos=50, ELside=r, ELdist=0.5, curvedepth=0](n0,n1){$1$}
\drawedge[ELpos=50, ELside=r, ELdist=0.5, curvedepth=0](n2,n3){$1/3$}
\drawedge[ELpos=50, ELside=r, ELdist=0.5, curvedepth=0](n2,n4){$2/3$}
\drawedge[ELpos=50, ELside=r, ELdist=0.5, curvedepth=0](n3,n5){$2/3$}
\end{picture}
\end{center}
  \caption{Markov chains $G$ and $\ov{G}$.}
  \label{fig:illus_1}
\end{figure*}

\begin{figure*}[t]
\begin{center}
\begin{picture}(85,25)(0,0)
\node[Nmarks=n](n0)(0,18){$s$}
\node[Nmarks=n](n1)(20,18){$t$}
\node[Nmarks=n](n2)(70,18){$s$}
\node[Nmarks=n](n3)(90,18){$t$}
\node[Nmarks=n](n4)(70,0){$s_1$}
\node[Nmarks=n](n5)(90,0){$t_1$}
\drawloop[loopangle=180](n4){$1$}
\drawloop[loopangle=0](n5){$1$}
\drawedge[ELpos=50, ELside=r, ELdist=0.5, curvedepth=-6](n0,n1){$1$}
\drawedge[ELpos=50, ELside=r, ELdist=0.5, curvedepth=-6](n1,n0){$1$}

\drawedge[ELpos=50, ELside=r, ELdist=0.5, curvedepth=-6](n2,n3){$1/3$}
\drawedge[ELpos=50, ELside=r, ELdist=0.5, curvedepth=-6](n3,n2){$1/3$}
\drawedge[ELpos=50, ELside=r, ELdist=0.5, curvedepth=0](n2,n4){$2/3$}
\drawedge[ELpos=50, ELside=r, ELdist=0.5, curvedepth=0](n3,n5){$2/3$}
\end{picture}
\end{center}
  \caption{Markov chains $G$ and $\ov{G}$.}
  \label{fig:illus_2}
\end{figure*}

We now consider another example to illustrate further. 
Consider the Markov chain $G$ and $\ov{G}$ in Fig~\ref{fig:illus_2}, 
where in $G$ it alternates 
between state $s$ and $t$, and the discount factor is $1/3$.
If we start at state $s$, the mean-discounted time at $t$ 
is given by 
\[
\frac{1/3^2 + 1/3^4 + 1/3^6 + \ldots}{1/3 + 1/3^2 + 1/3^3 +\ldots}
=\frac{1/9\cdot 9/8}{ 1/3 \cdot 3/2}=\frac{1}{4}.
\]
The probability to exit through $t_1$ in $\ov{G}$ in 2-steps is 
$1/3\cdot 2/3$, in 4-steps is $1/3^3\cdot 2/3$ and so on.
Hence the probability to exit through $t_1$ in $\ov{G}$ is 
\[
2/3 \cdot (1/3 + 1/3^3 + 1/3^5 + \ldots) = 2/3 \cdot 1/3 \cdot 9/8=1/4.
\]
The above examples show how the mean-discounted time in $G$ and the 
exit probability in $\ov{G}$ has the same value.
\qed
\end{example}

\begin{proof} {\em (of Lemma~\ref{lemm:poly}).}
We first write $h(x)$ as follows:
\[
h(x)= \sum_{i=1}^\ell a_i \cdot \prod_{j=1}^{n_i} x_{k_{ij}}, 
\]
where $\ell\in \Nats$, for all $i=1,2,\ldots,\ell$ we have $a_i\geq 0$,  
$n_i\leq n$, and $1 \leq k_{ij} \leq k$ for each $j=1,2,\ldots,n_i$.
By the hypothesis of the lemma, for all $i=1,2,\ldots, \ell$ we have 
\[
\frac{1}{(1+\vare)^n} \cdot \prod_{j=1}^{n_i} y'_{k_{ij}} \leq 
\prod_{j=1}^{n_i} y_{k_{ij}} \leq (1+\vare)^n \cdot \prod_{j=1}^{n_i} y'_{k_{ij}}.
\]
Since every $a_i\geq 0$, multiplying the above inequalities by $a_i$ and summing 
over $i=1,2,\ldots,\ell$ yields the desired result.
\qed
\end{proof}

\begin{proof} {\em (of Lemma~\ref{lemm_mc_multi}).}
We first observe that for a Markov chain $G$ we have 
$\Val(G,\MDT(\vdisc,r))(s)=\sum_{t\in S} r(t)\cdot \MT(s,t,G,\vdisc)$,
i.e., the value function for a state $s$ is obtained as the sum of 
the product of mean-discounted time of states and the rewards with $s$ as
the starting state.  
Hence by Lemma~\ref{lemm:poly} it follows that 
$\Val(G,\MDT(\vdisc,r))(s)$ can be expressed as a ratio $\frac{g_1(\cdot)}{g_2(\cdot)}$ 
of two polynomials of degree at most $|S|$ over $|S|^2$ variables.
Hence we have 
\[
\frac{\Val(G_1,\MDT(\vdisc,r))(s)}{\Val(G_2,\MDT(\vdisc,r))(s)} =
\frac{g_1(\trans)}{g_1(\trans')} \cdot \frac{g_2(\trans')}{g_2(\trans)} 
\]
Let $\vare=\dist(G_1,G_2)$.
By definition for all $s_1,s_2\in S$, if $s_2 \in \Supp(\trans(s_1))$, then 
we have both $\frac{\trans(s_1)(s_2)}{\trans'(s_1)(s_2)}$ and 
$\frac{\trans'(s_1)(s_2)}{\trans(s_1)(s_2)}$ are between $\frac{1}{1+\vare}$ and 
$1+\vare$.
It follows from Lemma~\ref{lemm:poly}, with $k=|S|^2$ that
\[
(1+\vare)^{-|S|} \leq \frac{g_i(\trans)}{g_i(\trans')} \leq (1+\vare)^{|S|}, \qquad \text{ for } i \in \set{1,2}.
\]
Thus we have
\[
(1+\vare)^{-2\cdot|S|} \leq 
\frac{g_1(\trans)}{g_1(\trans')}\cdot  
\frac{g_2(\trans')}{g_2(\trans)}   
\leq (1+\vare)^{2\cdot |S|}.
\] 
Hence we have
\[
(1+\vare)^{-2\cdot |S|} \leq 
\frac{\Val(G_1,\MDT(\vdisc,r))(s)}{\Val(G_2,\MDT(\vdisc,r))(s)} 
\leq (1+\vare)^{2\cdot |S|}
\]
We consider the case when 
$\Val(G_1,\MDT(\vdisc,r))(s)\geq \Val(G_2,\MDT(\vdisc,r))(s)$,
and the other case argument is symmetric. 
We also assume without loss of generality that $\Val(G_2,\MDT(\vdisc,r))(s)>0$. 
Otherwise if $\Val(G_2,\MDT(\vdisc,r))(s)=0$, since rewards are non-negative, 
it follows that no state with positive reward is reachable from $s$ both 
in $G_1$ and $G_2$ (because if they are reachable, then they are reachable with
positive probability and then the value is positive), and 
hence $\Val(G_1,\MDT(\vdisc,r))=\Val(G_2,\MDT(\vdisc,r))=0$ 
and the result of the lemma follows trivially.
Since we assume that $\Val(G_1,\MDT(\vdisc,r))(s)\geq \Val(G_2,\MDT(\vdisc,r))(s)$ and
$\Val(G_2,\MDT(\vdisc,r))(s)>0$, 
we have 
\[
\begin{array}{rcl}
|\Val(G_1,\MDT(\vdisc,r))(s) &- & \Val(G_2,\MDT(\vdisc,r))(s)| \\
& = & 
\displaystyle \Val(G_2,\MDT(\vdisc,r))(s)  \cdot \bigg( 
\frac{\Val(G_1,\MDT(\vdisc,r))(s)}{\Val(G_2,\MDT(\vdisc,r))(s)} -1 \bigg)\\[2ex] 
& \leq &  
\Val(G_2,\MDT(\vdisc,r))(s)  \cdot \big((1+ \vare)^{2\cdot|S|}-1\big) 
\end{array}
\]
Since the reward function is bounded by~1, it follows that 
$\Val(G_2,\MDT(\vdisc,r))(s) \leq 1$, and 
hence we have 
\[
|\Val(G_1,\MDT(\vdisc,r))(s) -\Val(G_2,\MDT(\vdisc,r))(s)| \leq (1+\dist(G_1,G_2))^{2\cdot|S|}-1.
\]
The desired result follows.
\qed
\end{proof}

\section{Missing proofs of Section~4}

\subsection{Details of Subsection 4.1}
We first show the desired result for MDPs and then extend to turn-based stochastic games.

\begin{theorem}\label{thrm_mdp}
Let $G_1$ be a player-1 MDP such that the minimum positive transition 
probability is $\eta>0$. 
The following assertions hold:
\begin{enumerate}
\item For all player-1 MDPs $G_2 \in \equivclass{G_1}$, 
for all parity objectives $\Phi$ and for all $s \in S$ we have 
\[
\begin{array}{rcl}
|\Val(G_1,\Phi)(s) - \Val(G_2,\Phi)(s)| 
& \leq &
(1+ \dist(G_1,G_2))^{2\cdot|S|}-1 \\[2ex]
& \leq &
\displaystyle
\bigg(1+\frac{\dista(G_1,G_2)}{\eta}\bigg)^{2\cdot |S|}-1
\end{array}
\]

\item For $\vare>0$, let 
$\beta \leq \frac{\eta}{2} \cdot\big((1+\frac{\vare}{2})^{\frac{1}{2\cdot|S|}}-1)$. 
For all $G_2 \in \equivclass{G_1}$ such that $\dista(G_1,G_2) \leq \beta$, for all 
parity objectives $\Phi$, every pure memoryless
optimal strategy $\stra_1$ in $G_1$ is an $\vare$-optimal strategy in $G_2$. 
In other words, for the interval $[0,\beta)$, every pure memoryless optimal
strategy in $G_1$ is an $\vare$-optimal strategy in all structurally equivalent
MDPs of $G_1$ such that the distance lies in the interval $[0,\beta)$.

\end{enumerate}
\end{theorem}
\begin{proof} We prove the two parts below.
\begin{enumerate}
\item Without loss of generality, let $\Val(G_1,\Phi)(s)\geq \Val(G_2,\Phi)(s)$.
Let $\stra_1$ be a pure memoryless optimal strategy in $G_1$ and such a strategy
exists by Theorem~\ref{thrm_lit1}.
Then we have the following inequality
\[
\begin{array}{rcl}
\Val(G_2,\Phi)(s) 
& \geq &
\Val(G_2 \restr \stra_1,\Phi)(s) \\[1ex] 
& \geq & 
\Val(G_1 \restr \stra_1,\Phi)(s) -
\big((1+ \dist(G_1,G_2))^{2\cdot|S|}-1\big) \\[1ex]
& = &
\Val(G_1,\Phi)(s) -
\big((1+ \dist(G_1,G_2))^{2\cdot|S|}-1\big) 
\end{array}
\]
The (in)equalities are obtained: the first inequality follows because
the value in $G_2$ is at least the value in $G_2$ obtained by fixing a
particular strategy (in this case $\stra_1)$; 
the second inequality is obtained by appying Theorem~\ref{thrm_mc_multi} 
on the structurally equivalent Markov chains $G_1\restr \stra_1$ and $G_2 \restr \stra_1$;
and the final equality follows since $\stra_1$ is an optimal strategy in $G_1$.
The desired result follows.

\item Let $G_2 \in \equivclass{G_1}$ such that $\dista(G_1,G_2)\leq \beta$. 
Let $\stra_1$ be any pure memoryless optimal strategy in $G_1$. 
Then we have the following inequality
\[
\begin{array}{rcl}
\Val(G_2 \restr \stra_1,\Phi)(s) 
& \geq & 
\Val(G_1 \restr \stra_1,\Phi)(s) - 
\big((1+ \dist(G_1,G_2))^{2\cdot|S|}-1\big) \\[1ex]
& = & \Val(G_1,\Phi)(s) - 
\big((1+ \dist(G_1,G_2))^{2\cdot|S|}-1\big) \\[1ex]
& \geq &  
\Val(G_2,\Phi)(s) - 
2\cdot \big((1+ \dist(G_1,G_2))^{2\cdot|S|}-1 \big).
\end{array}
\]
The first inequality is a consequence of Theorem~\ref{thrm_mc_multi} applied
on Markov chains $G_2 \restr \stra_1$ and $G_1 \restr \stra_1$; 
the equality follows from the fact $\stra_1$ is an optimal strategy 
in $G_1$; and the infinal equality follows by applying the 
result of part~1.
Hence to prove that $\stra_1$ is $\vare$-optimal in $G_2$
we need to show that 
\begin{eqnarray}\label{eq-diff1}
2\cdot \big((1+ \dist(G_1,G_2))^{2\cdot|S|}-1 \big) \leq \vare
\end{eqnarray}
We have 
\[
(1+ \dist(G_1,G_2)) \leq 
\bigg(1+\frac{\dista(G_1,G_2)}{\eta}\bigg); 
\]
the inequality follows from Proposition~\ref{prop_dist}. 
Hence to prove inequality (\ref{eq-diff1}) it suffices to show that 
\[
\bigg(1+\frac{\beta}{\eta}\bigg)^{2\cdot |S|}  \leq 1+\frac{\vare}{2}.
\]
Since $\beta \leq \frac{\eta}{2} \cdot\big((1+\frac{\vare}{2})^{\frac{1}{2\cdot|S|}}-1)$,
we obtain the desired inequality. 

\end{enumerate}
The desired result follows.
\qed
\end{proof}

\begin{proof} {\em (of Theorem~\ref{thrm_tb_stochastic}).}
The proof is essentially to repeat the proof of Theorem~\ref{thrm_mdp}: as in 
MDPs pure memoryless optimal strategies exist in turn-based stochastic games
with parity objectives (Theorem~\ref{thrm_lit1}); and once a pure memoryless
strategy is fixed in a turn-based stochastic game we obtain an MDP.
Since Theorem~\ref{thrm_mdp} extend the result of Theorem~\ref{thrm_mc_multi} 
from Markov chains to MDPs, the proof for the desired result follows 
by mimicking the proof of Theorem~\ref{thrm_mdp} and instead of using the
result of Theorem~\ref{thrm_mc_multi} for Markov chains using the result 
of Theorem~\ref{thrm_mdp} for MDPs.
\qed
\end{proof}

\subsection{Details of Subsection 4.2}

\begin{proof} {\em (of Lemma~\ref{lemm_mdp_multi}).}
The proof is essentially mimicking the proof of part(1) of
Theorem~\ref{thrm_mdp}.
Without loss of generality, let $\Val(G_1,\MDT(\vdisc,r))(s)\geq \Val(G_2,\MDT(\vdisc,r))(s)$.
Let $\stra_1$ be a pure memoryless optimal strategy in $G_1$ and such a strategy
exists by Theorem~\ref{thrm_lit1}.
Then we have the following inequality
\[
\begin{array}{rcl}
\Val(G_2,\MDT(\vdisc,r))(s) 
& \geq &
\Val(G_2 \restr \stra_1,\MDT(\vdisc,r))(s) \\[1ex]
& \geq & 
\Val(G_1 \restr \stra_1,\MDT(\vdisc,r))(s) -
\big((1+ \dist(G_1,G_2))^{2\cdot|S|}-1\big) \\[1ex]
& = &
\Val(G_1,\MDT(\vdisc,r))(s) -
\big((1+ \dist(G_1,G_2))^{2\cdot|S|}-1\big) 
\end{array}
\]
The (in)equalities are obtained: the first inequality follows because
the value in $G_2$ is at least the value in $G_2$ obtained by fixing a
particular strategy (in this case $\stra_1)$; 
the second inequality is obtained by appying Theorem~\ref{thrm_mc_multi} 
on the structurally equivalent Markov chains $G_1\restr \stra_1$ and $G_2 \restr \stra_1$;
and the final equality follows since $\stra_1$ is an optimal strategy in $G_1$.
The desired result follows.
\qed
\end{proof}

\begin{proof} {\em (of Lemma~\ref{lemm_conc_multi}).}
The proof is essentially mimicking the proof of Lemma~\ref{lemm_mdp_multi}.
Without loss of generality, let $\Val(G_1,\MDT(\vdisc,r))(s)\geq \Val(G_2,\MDT(\vdisc,r))(s)$.
Let $\stra_1$ be a randomized memoryless optimal strategy in $G_1$ and such a strategy
exists by Theorem~\ref{thrm_lit1}.
Then we have the following inequality
\[
\begin{array}{rcl}
\Val(G_2,\MDT(\vdisc,r))(s) 
& \geq &
\Val(G_2 \restr \stra_1,\MDT(\vdisc,r))(s) \\[1ex]
&\geq & 
\Val(G_1 \restr \stra_1,\MDT(\vdisc,r))(s) -
\big((1+ \dist(G_1,G_2))^{2\cdot|S|}-1\big) \\[1ex]
& = &
\Val(G_1,\MDT(\vdisc,r))(s) -
\big((1+ \dist(G_1,G_2))^{2\cdot|S|}-1\big) 
\end{array}
\]
The argument for the inequalities are exactly the same as in Lemma~\ref{lemm_mdp_multi}.
The desired result follows.
\qed
\end{proof}

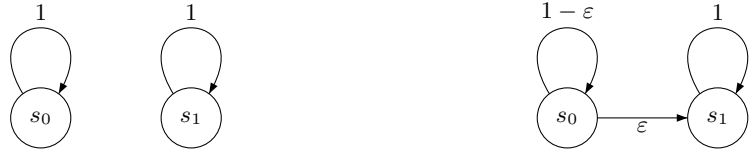
\begin{figure*}[t]
\begin{center}
\begin{picture}(85,25)(0,0)
\node[Nmarks=n](n0)(0,12){$s_0$}
\node[Nmarks=n](n1)(20,12){$s_1$}
\node[Nmarks=n](n2)(70,12){$s_0$}
\node[Nmarks=n](n3)(90,12){$s_1$}
\drawloop(n1){$1$}
\drawloop(n2){$1-\vare$}
\drawloop(n0){$1$}
\drawloop(n3){$1$}
\drawedge[ELpos=50, ELside=r, ELdist=0.5, curvedepth=0](n2,n3){$\vare$}
\end{picture}
\end{center}
  \caption{Markov chains $G_1$ and $G_2^\vare$ for Example~1.}
  \label{figure:buchi-lim}
\end{figure*}

\begin{example}[Asymptotically tight bound for small distances] We now show that 
the our quantitative bound for the value function difference is asymptotically optimal 
for small distances. 
Let us denote the absolute distance as $\vare$, and quantitative bound we 
obtain in Theorem~\ref{thrm_conc_multi} is $(1+\frac{\vare}{\eta-\vare})^{2\cdot|S|}-1$, and if $\vare$ is small 
($\vare<< \eta$ and $\vare$ close to zero), we obtain the following approximate bound 
\[
(1+\frac{\vare}{\eta-\vare})^{2\cdot|S|}-1 
\approx (1+\frac{\vare}{\eta})^{2\cdot |S|}-1 \approx 1 + 2\cdot |S| \cdot \frac{\vare}{\eta} -1 
= 2\cdot |S| \cdot \frac{\vare}{\eta}. 
\]
We now illustrate with an example (on structurally equivalent Markov chains) where the difference in the value 
function is $O(|S|\cdot \vare)$, for small $\vare$.
Consider the Markov chain defined on state space $S=\set{s_0,s_1,\ldots, s_{2n-1},s_{2n}}$
as follows: 
states $s_0$ and $s_{2n}$ are absorbing (states with self-loops of 
probability~1) and for a state $1\leq i \leq 2n-1$ we have 
\[
\trans(s_i)(s_{i-1})=\frac{1}{2}+\vare; \quad \trans(s_i)(s_{i+1})=\frac{1}{2}-\vare;
\] 
i.e., we have a Markov chain defined on a line from $0$ to $2n$ (with $0$ and $2n$ absorbing
states) and the chain moves towards $0$ with probability $\frac{1}{2}+\vare$ and towards
$2n$ with probability $\frac{1}{2}-\vare$ (see Fig~\ref{figure:asym}). 
Our goal is to estimate the probability to reach the state $s_0$, and let $v_i$ denote 
the probability to reach $s_0$ from the starting state $s_i$.
Then we have the following simple recurrence for $1\leq i \leq 2n-1$
\[
v_i =(\frac{1}{2}+\vare)\cdot v_{i-1} + (\frac{1}{2}-\vare)\cdot v_{i+1};  
\] 
and $v_0=1$ and $v_{2n}=0$.
We will consider $\vare\geq 0$ such that $\vare$ is very small and hence higher order 
terms (like $\vare^2$) can be ignored. 
We claim that the values $v_i$ can be expressed as the following recurrence:
$v_{i+1}=(\frac{1}{2}+\vare) \cdot c_i \cdot v_i$, where $c_i=\frac{4}{4-c_{i+1}}$. 
The proof is by induction and is shown below:
\[
\begin{array}{rcl}
v_i & = & 
(\frac{1}{2}+\vare)\cdot v_{i-1} + (\frac{1}{2}-\vare)\cdot v_{i+1} \\[1ex]
 & = & 
(\frac{1}{2}+\vare)\cdot v_{i-1} + (\frac{1}{2}-\vare)\cdot (\frac{1}{2} +\vare)\cdot c_i \cdot v_{i} \quad \text{(by inductive hypothesis $v_{i+1}=
(\frac{1}{2}+\vare)\cdot c_i \cdot v_i$)}\\[1ex]
&= & (\frac{1}{2}+\vare)\cdot v_{i-1} + (\frac{1}{4}-\vare^2) \cdot c_i \cdot v_{i}\\[1ex] 
&= & (\frac{1}{2}+\vare)\cdot v_{i-1} + \frac{1}{4} \cdot c_i \cdot v_{i} \qquad \text{(ignoring $\vare^2$)}
\end{array}
\]
It follows that $v_i=(\frac{1}{2}+\vare)\cdot \frac{4}{4-c_i}\cdot v_{i-1}= 
(\frac{1}{2}+\vare)\cdot c_{i-1} \cdot v_{i-1}$.
Hence we have 
\[
\begin{array}{rcl}
v_1 
& = & 
(\frac{1}{2}+\vare)\cdot v_0 + 
(\frac{1}{2}-\vare)\cdot v_2 \\[1ex]
& = & 
(\frac{1}{2}+\vare)\cdot 1 + 
(\frac{1}{2}-\vare)\cdot (\frac{1}{2}+\vare)\cdot c_1 \cdot v_1 \\[1ex]
& = & 
(\frac{1}{2}+\vare) + (\frac{1}{4}-\vare^2)\cdot c_1 \cdot v_1 \\[1ex] 
& = & (\frac{1}{2}+\vare) + \frac{1}{4}\cdot c_1 \cdot v_1 \qquad
\text{(ignoring $\vare^2$)}
\end{array}
\] 
Thus we obtain that $v_1=\frac{4}{4-c_1}\cdot (\frac{1}{2}+\vare)$.
Then we have $v_2=(\frac{1}{2}+\vare) \cdot c_1 \cdot v_1 = 
\frac{4}{4-c_1}\cdot c_1 \cdot (\frac{1}{2}+\vare)^2$ and then 
$v_3= \frac{4}{4-c_1}\cdot c_1 \cdot c_2 \cdot (\frac{1}{2}+\vare)^3$ and so on.
Finally we obtain $v_n$ as follows: 
$v_n= \frac{4}{4-c_1}\cdot c_1 \cdot c_2 \cdots c_{n-1} \cdot (\frac{1}{2}+\vare)^n$.
Observe that for the Markov chain with $\vare=0$, the states $s_0$ and $s_{2n}$ are
the recurrent states, and since the chain is symmetric from $s_n$ (with $\vare=0$)
the probability to reach $s_{2n}$ and $s_0$ must be equal and hence is $\frac{1}{2}$.
It follows that we must have 
$\frac{4}{4-c_1}\cdot c_1 \cdot c_2 \cdots c_{n-1} =2^{n-1}$.
Hence we have that for $\vare>0$, but very small, 
$v_n \approx \frac{1}{2} + n\cdot \vare$. 
Thus the difference with the value function when $\vare=0$ as compared to when 
$\vare>0$ but very small is $n\cdot \vare =O(|S|\cdot \vare)$.
Also observe that the Markov chain obtained for $\vare=0$ and $\frac{1}{2} > 
\vare >0$ are structurally equivalent. Thus the desired result follows.
\begin{figure*}[t]
\begin{center}
\begin{picture}(85,25)(10,0)
\node[Nmarks=n](n0)(0,12){$s_0$}
\node[Nmarks=n](n1)(20,12){$s_1$}
\node[Nmarks=n](n2)(40,12){$s_2$}
\node[Nmarks=n](n3)(70,12){$s_{2n-2}$}
\node[Nmarks=n](n4)(90,12){$s_{2n-1}$}
\node[Nmarks=n](n5)(110,12){$s_{2n}$}

\put(50,12){$\cdots$$\cdots$}

\drawloop[loopangle=180](n0){$1$}
\drawloop[loopangle=0](n5){$1$}
\drawedge[ELpos=50, ELside=l, ELdist=0.5, curvedepth=3](n1,n0){$\frac{1}{2}+\vare$}
\drawedge[ELpos=50, ELside=l, ELdist=0.5, curvedepth=3](n1,n2){$\frac{1}{2}-\vare$}
\drawedge[ELpos=50, ELside=l, ELdist=0.5, curvedepth=3](n2,n1){$\frac{1}{2}+\vare$}

\drawedge[ELpos=50, ELside=l, ELdist=0.5, curvedepth=3](n4,n5){$\frac{1}{2}-\vare$}
\drawedge[ELpos=50, ELside=l, ELdist=0.5, curvedepth=3](n4,n3){$\frac{1}{2}+\vare$}
\drawedge[ELpos=50, ELside=l, ELdist=0.5, curvedepth=3](n3,n4){$\frac{1}{2}-\vare$}

\end{picture}
\end{center}
  \caption{Markov chains for Example~2.}
  \label{figure:asym}
\end{figure*}
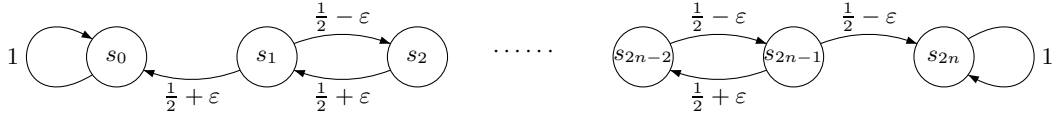
\qed
\end{example}

\end{document}